\documentclass[preprint,12pt]{mystyle}
\usepackage{graphicx}

\usepackage{subfigure}
\usepackage{psfrag}
\graphicspath{{figs/}}

\usepackage[centertags]{amsmath}
\usepackage{amsfonts}
\usepackage{color}
\usepackage{amssymb}
\usepackage{amsthm}
\usepackage{newlfont}
\usepackage[applemac]{inputenc}

\newtheorem{lemma}{Lemma}
\newtheorem{theorem}[lemma]{Theorem}

\newproof{pf}{Proof}

\def   \qed            {\hfill $\Box$\par\bigskip}
\def\Box{\hbox{\vrule width6pt height6pt depth0pt}}

{\newdimen\quadamount\quadamount=1.5em%
\def\TAB##1{\\\hskip##1\quadamount\relax}%
\medskip\sf\begin{tabular}{l}}%
{\end{tabular}\medskip}%

\journal{DAM Special Issue: Combinatorial Optimization}

\begin{document}
\sloppy
\begin{frontmatter}

%
%
%
%

\title{Complexity of colouring problems restricted to unichord-free and $\{$square,unichord$\}$-free graphs\tnoteref{t1}}
\tnotetext[t1]{February 3, 2012. An extended abstract of this paper was presented at ISCO 2010, the International Symposium on Combinatorial Optimization, and appeared in Electronic Notes in Discrete Mathematics 36 (2010) 671--678.}

\author[rm]{Raphael C. S. Machado}
\ead{rcmachado@inmetro.gov.br}

\author[cf]{Celina M. H. de Figueiredo}
\ead{celina@cos.ufrj.br}

\author[nt]{Nicolas Trotignon}
\ead{nicolas.trotignon@ens-lyon.fr}

\address[rm]{Inmetro, Rio de Janeiro}
\address[cf]{COPPE, UFRJ, Rio de Janeiro}
\address[nt]{CNRS, LIP, ENS Lyon, UCBL, INRIA, Universit\'e de Lyon}


\begin{abstract}
  A \emph{unichord} in a graph is an edge that is the unique chord of
  a cycle. A \emph{square} is an induced cycle on four vertices.
  A graph is \emph{unichord-free} if none of its edges is a
  unichord. We give a slight restatement of a known structure theorem
  for unichord-free graphs and use it to show that, with the only exception 
  of the complete graph $K_4$, every square-free, unichord-free
  graph of maximum degree~3 can be total-coloured with four colours.
  Our proof can be turned into a
  polynomial time algorithm that actually outputs the colouring.
  This settles the class of  square-free, unichord-free graphs
  as a class for which edge-colouring is NP-complete but
  total-colouring is polynomial.  

\end{abstract}

\end{frontmatter}


{
\typeout{:?0000} 
\typeout{:?1111} 
}
\def\baselinestretch{1}
\typeout{Introduction}

\section{Introduction}
\label{s:introduction}
In the present paper, we deal with simple connected graphs. A graph $G$
has vertex set $V(G)$ and edge set $E(G)$. An \emph{element} of $G$ is
one of its vertices or edges and the set of elements of $G$ is denoted
by $S(G):=V(G)\cup E(G)$. Two vertices $u,v\in V(G)$ are
\emph{adjacent} if $uv\in E(G)$; two edges $e_{1},e_{2}\in E(G)$ are
\emph{adjacent} if they share a common endvertex; a vertex $u$ and an
edge~$e$ are \emph{incident} if $u$ is an endvertex of $e$. For a
graph $G=(V,E)$ and $V' \subseteq V$,  $G[V']$ denotes the
subgraph of $G$ induced by $V'$. The \emph{degree} of a vertex $v$ in
$G$ is the number of edges of $G$ incident to $v$.  We use the
standard notation of $K_{n}$, $K_{m,n}$ and $C_{n}$ for complete
graphs, complete bipartite graphs and cycle-graphs, respectively.

An \emph{edge-colouring} is an association of colours to the edges of a graph in such a way that no adjacent edges receive the same colour. The \emph{chromatic index} of a graph $G$, denoted $\chi'(G)$, is the least number of colours sufficient to edge-colour this graph. Clearly, $\chi'(G)\geq\Delta(G)$, where $\Delta(G)$ denotes the maximum degree of a vertex in $G$. Vizing's theorem~\cite{Vizing} states that every graph $G$ can be edge-coloured with $\Delta(G)+1$ colours. By Vizing's theorem only two values are possible for the chromatic index of a graph: $\chi'(G)=\Delta(G)$~or~$\Delta(G)+1$. If a graph~$G$ has chromatic index $\Delta(G)$, then $G$ is said to be \emph{Class~1}; if $G$ has chromatic index $\Delta(G)+1$, then $G$ is said to be \emph{Class~2}. 

A \emph{total-colouring} is an association of colours to the elements of a graph in such a way that no adjacent or incident elements receive the same colour. The \emph{total chromatic number} of a graph $G$, denoted $\chi_{T}(G)$, is the least number of colours sufficient to total-colour this graph. Clearly, $\chi_{T}(G)\geq\Delta(G)+1$. The Total Colouring Conjecture (TCC) states that every graph $G$ can be total-coloured with $\Delta(G)+2$ colours. By the TCC only two values would be possible for the total chromatic number of a graph: $\chi_{T}(G)=\Delta(G)+1$~or~$\Delta(G)+2$. If a graph~$G$ has total chromatic number $\Delta(G)+1$, then $G$ is said to be \emph{Type~1}; if $G$ has total chromatic number $\Delta(G)+2$, then $G$ is said to be \emph{Type~2}. The TCC has been verified in restricted cases, such as graphs with maximum degree~$\Delta\leq5$~\cite{Kostochka1,Kostochka2,Rosenfeld,Vijayaditya}, but the general problem is open since 1964, exposing how challenging the problem of total-colouring~is.

It is NP-complete to determine whether the total chromatic number of a
graph $G$ is $\Delta(G)+1$~\cite{SanchezArroyo1,SanchezArroyo}. Remark that the
original NP-completeness proof was a reduction from the edge-colouring
problem, suggesting that, for most graph classes, total-colouring
would be harder than edge-colouring. The present paper presents the
first example of an unexpected graph class for which edge-colouring is NP-complete
while total-colouring is polynomial. For a discussion on the search of 
complexity separating classes for edge-colouring and total-colouring 
please refer to~\cite{BipJBCS}.

A \emph{square} is an induced cycle on four vertices. A \emph{unichord} is an edge that is the
unique chord of a cycle in the graph.  In the present work, we consider total-colouring restricted to
$\{$square,unichord$\}$-free graphs --- that is, graphs that do not contain (as an induced
subgraph) a cycle with a unique chord nor a square. The class of unichord-free
graphs was studied by Trotignon and
Vu{\v s}kovi{\'c}~\cite{tv}. They give a structure theorem for the
class, and use it to develop algorithms for recognition and
vertex-colouring. Basically, this structure result states that every
unichord-free graph can be built starting from a restricted set of
basic graphs and applying a series of known ``gluing'' operations. The
following results are obtained in~\cite{tv} for unichord-free graphs:
an $O(nm)$ recognition algorithm, an $O(nm)$ algorithm for optimal
vertex-colouring, an $O(n+m)$ algorithm for maximum clique, and the
NP-completeness of the maximum stable set problem.

Machado, Figueiredo and Vu{\v s}kovi{\'c}~\cite{Edge-tv} investigated
whether the structure results of~\cite{tv} could be applied to obtain
a polynomial-time algorithm for the edge-colouring problem restricted
to unichord-free graphs. The authors obtained a negative answer by
establishing the NP-completeness of the edge-colouring problem
restricted to unichord-free graphs. The authors investigated also the
complexity of the edge-colouring in the subclass of
$\{$square,unichord$\}$-free graphs. The class of
$\{$square,unichord$\}$-free graphs can be viewed as the class of
graphs that can be constructed from the same set of basic graphs, but
using one less operation (the so-called \emph{1-join operation} is
forbidden). For $\{$square,unichord$\}$-free graphs, an interesting
dichotomy is proved in~\cite{Edge-tv}: if the maximum degree is not~3,
the edge-colouring problem is polynomial, while for inputs with
maximum degree~3, the problem is NP-complete.

It is a natural step to investigate the complexity of total-colouring
restricted to classes for which the complexity of edge-colouring is
already established. This approach is observed, for example, in the
classes of outerplanar graphs~\cite{Zhang},
series-parallel graphs~\cite{WangPang}, and some subclasses of
planar graphs~\cite{WeifanWang} and join
graphs~\cite{OneKindJoin,OutroJoin}. One important motivation for this approach is the search for
``separating'' classes, that are classes for which the complexities of
edge-colouring and total-colouring differ. We must mention that all
previously known separating classes, in this sense, are classes for
which edge-colouring is polynomial and total-colouring is NP-complete,
such as the case of bipartite graphs. In other words, there is no
known example of a class for which edge-colouring is NP-complete and
total-colouring is polynomial, an evidence that total-colouring might be
``harder'' than edge-colouring.

Considering the recent interest in colouring problems restricted to
unichord-free and $\{$square,unichord$\}$-free graphs, specially the
results~\cite{Total-tv} on total-colouring $\{$square,unichord$\}$-free graphs of
maximum degree at least~4, it is natural to investigate the remaining
case of total-colouring restricted to
$\{$square,unichord$\}$-free graphs of maximum degree~3. In the
present work, we prove that, except for the complete graph $K_{4}$,
every $\{$square,unichord$\}$-free graph of maximum degree~3 is
Type~1.  Our proof can easily be turned into a polynomial time
algorithm that outputs the colouring whose existence is proved (we omit
the details of the implementation).  Table~\ref{t:tabela} summarizes
the current status of colouring problems restricted to unichord-free
and $\{$square,unichord$\}$-free graphs.

\begin{center}
\begin{table}[t]
\begin{tabular}{|c|c|c|c|}
\hline
Problem $\setminus$ Class & unichord-free & $\{$sq.,un.$\}$-free, $\Delta\geq4$ & $\{$sq.,un.$\}$-free, $\Delta=3$\\
\hline\hline
vertex-colouring & Polynomial~\cite{tv} & Polynomial~\cite{tv} & Polynomial~\cite{tv} \\
edge-colouring & NP-complete~\cite{Edge-tv} & Polynomial~\cite{Edge-tv} & NP-complete~\cite{Edge-tv} \\
total-colouring & NP-complete~\cite{Total-tv} & Polynomial~\cite{Total-tv} & Polynomial$^{*}$ \\
  \hline
\end{tabular}
\caption{Computational complexity of colouring problems restricted to unichord-free and to $\{$square,unichord$\}$-free graphs --- star indicates result established in the present paper.}
\label{t:tabela}
\end{table}
\end{center}

Observe in Table~\ref{t:tabela} the interesting degree dichotomy with respect to
edge-colouring $\{$square,unichord$\}$-free graphs. Since the
technique used in~\cite{Total-tv} to total-colour
$\{$square,unichord$\}$-free graphs could only be applied to the case
of maximum degree at least~4, a similar dichotomy could be expected
for the total-colouring problem. Surprisingly, we establish in the
present work that such dichotomy does not exist. It is additionally
interesting to note that different approaches were needed to solve the
total-colouring problem in the cases $\Delta\geq4$ and $\Delta=3$.
Note that a natural subclass of unichord-free graphs is the class of
\emph{chordless graphs}, that are the graphs where all cycles are
chordless.  For these graphs, we have proved that edge- and total-colouring are
all polynomially solvable, with no restriction on the degree and the
presence of squares~\cite{Chordless}.

In Section~\ref{s:structural}, we recall the structure theorem for
unichord-free graphs.  We restate it in a slightly different form that
is well-fit to our goal of total-colouring. In
Section~\ref{s:total-colouring}, we prove the main result of the paper
that non-complete $\{$square, unichord$\}$-free graphs with maximum
degree~3 are Type~1.

\section{Decomposing unichord-free graphs}
\label{s:structural}

We revisit the decomposition result for unichord-free~\cite{tv}
graphs, stating it in a new form that will be suitable for
total-colouring.

\subsection{Decomposition theorem}

The \emph{Petersen graph} is the cubic graph on vertices $\{a_1,
\dots, a_5, b_1, \dots, b_5\}$ so that both $a_1a_2a_3a_4a_5a_1$ and
$b_1b_2b_3b_4b_5b_1$ are chordless cycles, and such that the only
edges between some $a_i$ and some $b_i$ are $a_1b_1$, $a_2b_4$,
$a_3b_2$, $a_4b_5$, $a_5b_3$.  Figure~\ref{f:basicPetersen} exhibits a
(total-coloured) graph isomorphic to the Petersen graph.  We denote by
$P$ the Petersen graph and by $P^{*}$ the graph obtained from $P$ by
the removal of one vertex. Observe that $P$ is unichord-free.

The \emph{Heawood graph} is the cubic bipartite graph on vertices
$\{a_1, \dots , a_{14}\}$ so that $a_1a_2\dots a_{14}a_1$ is a cycle,
and such that the only other edges are $a_{1}a_{10}$, $a_{2}a_{7}$,
$a_{3}a_{12}$, $a_{4}a_{9}$, $a_{5}a_{14}$, $a_{6}a_{11}$,
$a_{8}a_{13}$.  Figure~\ref{f:basicHeawood} exhibits a
(total-coloured) graph isomorphic to the Heawood graph.  The
Hamiltonian cycle from the definition is shown in bold edges.  We
denote by $H$ the Heawood graph and by $H^{*}$ the graph obtained from
$H$ by the removal of one vertex. Observe that $H$ is unichord-free.

It essential for understanding what follows to notice that the
Petersen and Heawood graphs are both vertex-transitive.  It is also
helpful to know their most classical embeddings, as shown for instance
in \cite{tv}.

\begin{figure}[t]
     \centering
     \subfigure[Petersen graph.]{
          \label{f:basicPetersen}
          \includegraphics[width=140pt]{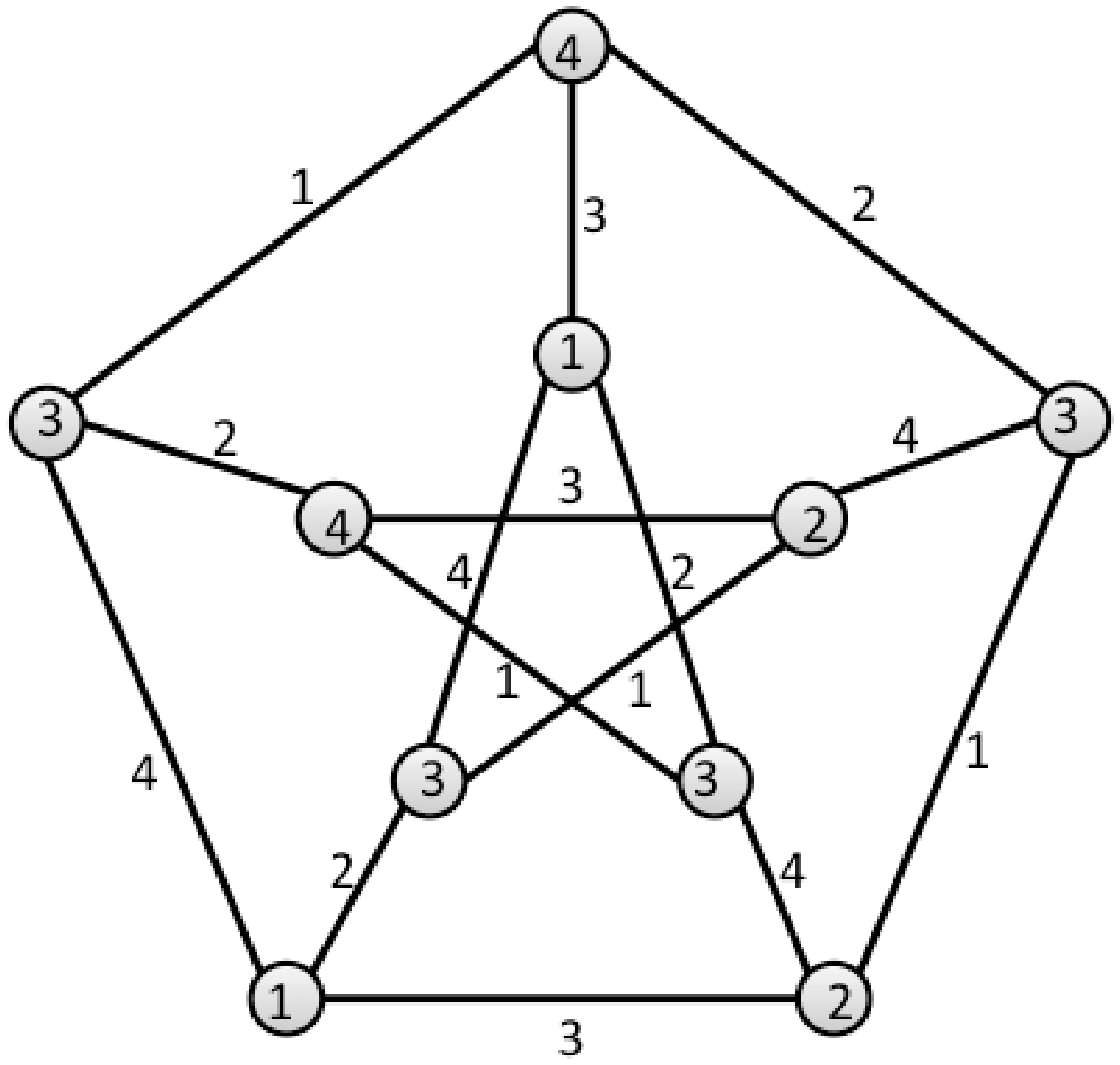}
}
          \hspace{20pt}
     \subfigure[Heawood graph.]{
          \label{f:basicHeawood}
          \includegraphics[width=180pt]{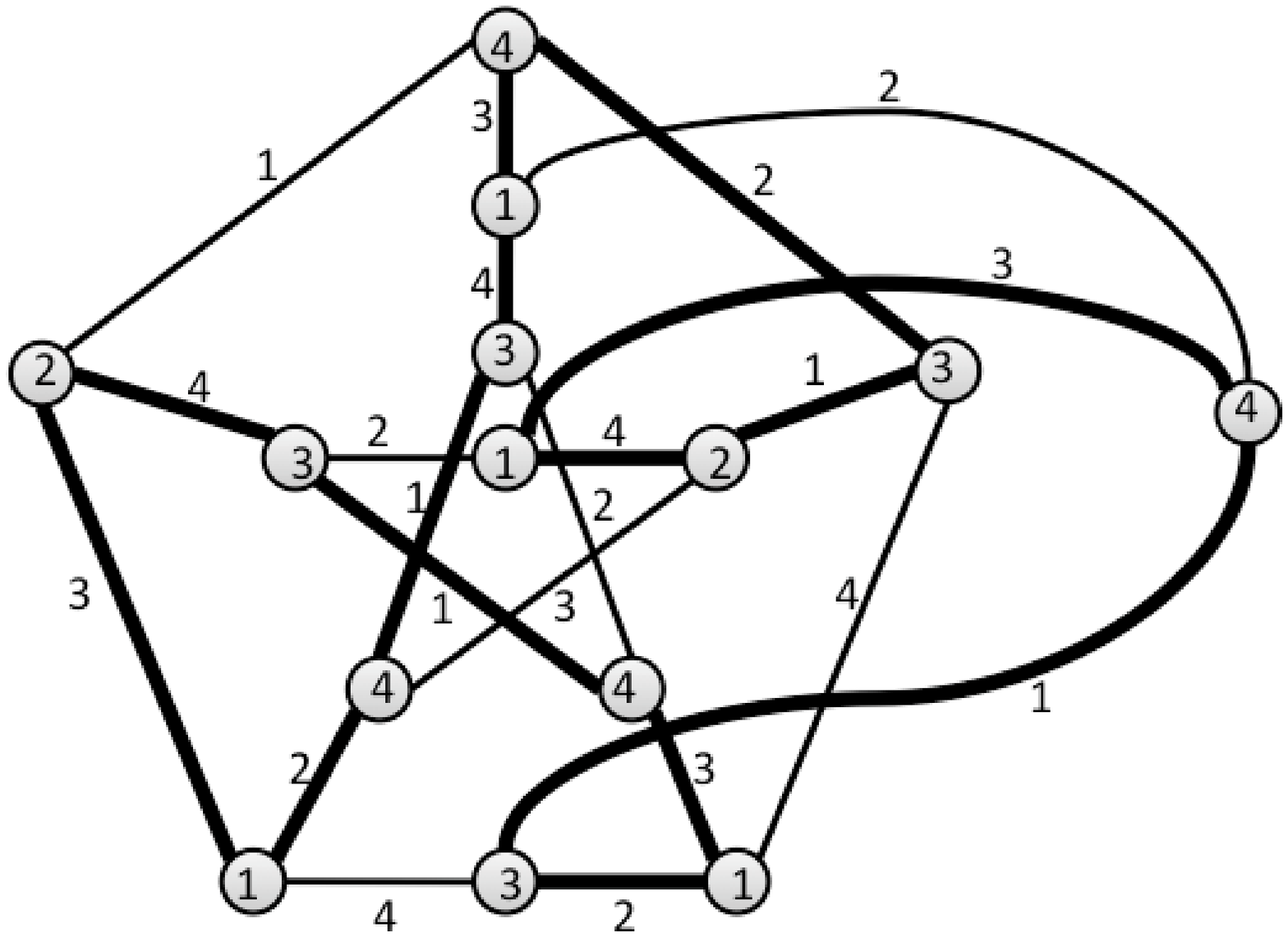}
}
    \caption{4-total-colourings of the Petersen graph and of the Heawood graph.}
    \label{f:colouring-basic}
\end{figure}

A graph is \emph{strongly 2-bipartite} if it is square-free and
bipartite with bipartition $(X,Y)$ where every vertex in $X$ has
degree~2 and every vertex in $Y$ has degree at least~3. A strongly
2-bipartite graph is unichord-free because any chord of a cycle is an
edge between two vertices of degree at least three, so every
cycle in a strongly 2-bipartite graph is chordless.

A \emph{cutset} $S$ of a connected graph $G$ is a set of vertices or 
a set of edges whose removal disconnects $G$.  A decomposition
of a graph is the systematic 
removal of cutsets to obtain smaller graphs 
(by adding vertices and edges to connected components of $G\setminus S$),
called the {\em blocks of decomposition}, repeating this until a set of basic (undecomposable) graphs is obtained.
The goal
of decomposing a graph is trying to solve a problem on the original
graph by combining the solutions on the blocks.  The following cutsets
are used in the decomposition theorem of Trotignon and Vu{\v
  s}kovi{\'c} for unichord-free graphs~\cite{tv}:

\begin{itemize}
\item A \emph{1-cutset} of a connected graph $G=(V,E)$ is a vertex $v$
  such that $V$ can be partitioned into sets $X$, $Y$ and $\{ v \}$,
  so that there is no edge between $X$ and $Y$.  We say that $(X,
  Y,v)$ is a \emph{split} of this 1-cutset.

\item A \emph {special 2-cutset} of a connected graph $G=(V,E)$ is a
  pair of non-adjacent vertices $a, b$, both of degree at least three,
  such that $V$ can be partitioned into sets $X$, $Y$ and $\{ a,b \}$
  so that: $|X|\geq 2$, $|Y| \geq 2$; there is no edge between $X$ and
  $Y$, and both $G[X \cup \{ a,b \}]$ and $G[Y \cup \{ a,b \}]$
  contain an $ab$-path.  We say that $(X, Y, a, b)$ is a \emph{split}
  of this special 2-cutset.  Note that in~\cite{tv},
  \emph{special} 2-cutsets are called \emph{proper} 2-cutsets.  We
  apologize for changing the terminology, but we find it more convenient
  to keep the ``proper 2-cutset'' for the restatement given in Section~\ref{s:revisited}.

\item A \emph{proper 1-join} of a graph $G=(V,E)$ is a partition of
  $V$ into sets $X$ and $Y$ such that there exist sets $A \subseteq X$
  and $B \subseteq Y$ so that: $|A|\geq 2$, $|B| \geq 2$; $A$ and $B$
  are stable sets; there are all possible edges between $A$ and $B$;
  there is no other edge between $X$ and $Y$.  We say that $(X, Y, A,
  B)$ is a \emph{split} of this proper 1-join.
\end{itemize}

We are now ready to state a decomposition result for unichord-free graphs.

\begin{theorem}
  \label{th:1}
  {\em (Trotignon and ~Vu{\v s}kovi{\'c}~\cite{tv})} If $G$ is a
  connected unichord-free graph, then either $G$ is a complete graph,
  or a cycle, or a strongly 2-bipartite graph, or an induced subgraph
  of the Petersen graph, or an induced subgraph of the Heawood graph,
  or $G$ has a 1-cutset, a special 2-cutset, or a proper 1-join.
\end{theorem}

The decomposition blocks with respect to 1-cutsets and special
2-cutsets are defined below (we do not use here the blocks with
respect to proper 1-joins).

The \emph{block} $G_X$ (resp.\ $G_Y$) of a graph $G$ with respect to a
1-cutset with split $(X, Y, v)$ is $G[X\cup \{ v \} ]$ (resp.\ $G[Y\cup
\{ v\} ]$).

The \emph{blocks} $G_X$ and $G_Y$ of a graph $G$ with respect to a
special 2-cutset with split $(X, Y, a, b)$ are defined as follows.  If
there exists a vertex $c$ of $G$ such that $N_G(c)=\{ a,b \}$, then
let $G_X=G[X \cup \{ a,b,c \}]$ and $G_Y=G[Y \cup \{ a,b,c \}]$.
Otherwise, block $G_X$ (resp.\ $G_Y$) is the graph obtained by taking
$G[X\cup \{ a,b\}]$ (resp.\ $G[Y\cup \{ a,b\}]$) and adding a new
vertex $c$ adjacent to $a, b$.  Vertex $c$ is called the \emph{marker}
of the block $G_X$ (resp.\ $G_Y$).

The decomposition blocks of a unichord-free graph with respect to
1-cutsets and special 2-cutsets are constructed in such a way that
they remain unichord-free~\cite{tv}. Additionally, the decomposition
blocks of a $\{$square,unichord$\}$-free graph are themselves
$\{$square,unichord$\}$-free~\cite{Edge-tv}.

\subsection{Restated decomposition theorem}
\label{s:revisited}

We derive a restatement of the decomposition result for unichord-free
graphs that fits better to our total-colouring purposes (possibly
for other purposes as well).  The proposed restatement needs the
following notion of \emph{proper} 2-cutset.

A \emph {proper 2-cutset} of a graph $G=(V,E)$ is a pair of
non-adjacent vertices $a, b$ such that $V$ can be partitioned into
sets $X$, $Y$ and $\{ a,b \}$ so that: $|X|\geq 2$, $|Y| \geq 2$;
there is no edge between $X$ and $Y$; and both $G[X \cup \{ a,b \}]$
and $G[Y \cup \{ a,b \}]$ contain an $ab$-path but none of them is an
$ab$-path.  We say that $(X, Y, a, b)$ is a \emph{split} of this
proper 2-cutset.  Note that a proper 2-cutset is a particular kind of
special 2-cutset (so we may still use the notion of block of
decomposition as defined previously).

A \emph{branch vertex} of a graph is any vertex of degree at least 3,
and we call \emph{branch} any path whose endvertices are branch
vertices and whose internal vertices are not. Observe that a
2-connected graph that is not a cycle can be edge-wise partitioned into
its branches.  A graph is \emph{sparse} if its branch vertices form a
stable set (so every edge is incident to at least one vertex of degree
at most~2).  Note that every strongly 2-bipartite graph is sparse.
The \emph{reduced graph} of a 2-connected graph $G$ that is not a
cycle is the graph obtained from $G$ by contracting every branch of
length at least 3 into a branch of length 2.

A \emph{2-extension} of a graph $G$ is any graph obtained by
(first) deleting vertices from $G$ and (second) subdividing edges
incident to at least one vertex of degree~2. Note that the
unichord-free class is closed under taking 2-extensions.

We can now restate the decomposition theorem of Trotignon and Vu{\v
  s}kovi{\'c} for unichord-free graphs.  The difference with the
original theorem is that we use the more precise ``proper'' 2-cutset
instead of the ``special'' 2-cutset.  The price to pay for that is an
extension of the basic classes: instead of the strongly 2-bipartite
graphs, induced subgraphs of Petersen and induced subgraphs of
Heawood, we have to use the less precise sparse graphs, 2-extensions of
Petersen and 2-extensions of Heawood respectively.  Another small
difference is that cycles do not form a separate basic class anymore since
they are sparse.

\begin{theorem}
  \label{t:newdecomposition}
  If $G$ is a connected unichord-free graph, then either $G$ is a
  complete graph, or a sparse graph, or a 2-extension of Petersen
  or Heawood graph, or $G$ has a 1-cutset,
  a proper 2-cutset, or a proper 1-join.
\end{theorem}

\begin{proof}
  We may assume that $G$ is not a cycle and is 2-connected (for
  otherwise, it is sparse or has a 1-cutset). Let $G'$ be the reduced
  graph of~$G$. Observe that $G$ can be obtained from $G'$ by
  subdividing edges incident to at least one vertex of degree 2, that $G'$ is not a cycle, and
  that $G'$ is 2-connected.  Also, $G'$ is unichord-free, since
  contracting a path of length at least~3 into a path of length~2
  does not create nor destroy chords of cycles. 

  We apply the decomposition Theorem~\ref{th:1} to $G'$.  If $G'$ is a
  complete graph on at least~4 vertices, then in fact $G = G'$, so we
  are done.  If $G'$ is 
  a strongly 2-bipartite graph then $G$ is sparse.  If $G'$
  is an induced subgraph of the Petersen graph or of the Heawood
  graph, then $G$ is a 2-extension of Petersen or
  Heawood graph.  If $G'$ has a special 2-cutset $\{a,b\}$ with split $(X,Y,a,b)$, 
  then $\{a,b\}$ is a proper 2-cutset of $G$ (since $G'$ is reduced, no
  side $G[X\cup\{a,b\}]$ or $G[Y\cup\{a,b\}]$ of a special 2-cutset in $G'$ can be a path,
  because this would imply that $|X|=1$ or $|Y|=1$).

  Finally consider the case where $G'$ has a proper 1-join with split
  $(X,Y,A,B)$.  Suppose that $(X, Y, A, B)$ is chosen so that the
  number $k$ of vertices of degree 2 in $A \cup B$ is minimal.  If
  $k=0$, then all vertices in $A \cup B$ have degree at least~3, so
  $G$ is obtained from $G'$ by subdividing edges with both ends in $X$
  or both ends in $Y$, so that the edges between $A$ and $B$ still
  form a proper 1-join in $G$.  Hence, we may assume that in $G'$,
  there is a vertex $u\in A$ (up to symmetry) of degree~2.  It follows
  that $|B| = 2$, say $B = \{v, w\}$ and $B$ is the neighborhood of
  $u$ in $G'$.  If $A \geq 3$, then $(X \setminus \{u\}, Y \cup \{u\},
  A \setminus \{u\}, B)$ is a split of a proper 1-join of $G'$ that
  contradicts the minimality of $k$.  So, $|A| = 2$, say $A = \{u,
  u'\}$.  Since $G$ is 2-connected, $u'$ cannot be a 1-cutset, so in
  fact $X = A = \{u, u'\}$.  If $|Y| \geq 4$, then $(X, Y \setminus
  \{v, w\}, v, w)$ is a split of a special 2-cuset of $G'$, so we are
  done as in the previous paragraph.  Hence, $|Y| \leq 3$.  Now, all
  vertices in $G'$ have degree at most 2, except possibly $v$ and $w$
  that are non-adjacent.  It follows that $G'$ is sparse, and so is
  $G$ (in fact, $4 \leq |V(G')| \leq 5$, and $G'$ is isomorphic to
   the square or to $K_{2, 3}$).  \qed
\end{proof}

A more precise theorem is obtained for 2-connected square-free graphs.

\begin{theorem}
  \label{t:newdecompositionNoS}
  If $G$ is a 2-connected \{square, unichord\}-free graph, then either~$G$ 
  is a complete graph, or a sparse graph, or a 2-extension of 
  Petersen or Heawood graph,
  or has a proper 2-cutset.
\end{theorem}

\begin{proof}
  Follows directly from Theorem~\ref{t:newdecomposition} because a
  1-join cannot occur in a square-free graph (if a graph has a 1-join,
  it must contain a square, formed by any two vertices from $A$ and
  two vertices from $B$).\qed  
\end{proof}

The following lemma restates the extremal decomposition of~\cite{tv}.

\begin{lemma}
\label{l:newextremal}
Let $G$ be a 2-connected $\{$square,unichord$\}$-free graph and let
$(X, Y, a, b)$ be a split of a proper 2-cutset of $G$ such that $|X|$
is minimum among all possible such splits.  Then $a$ and $b$ both have
at least two neighbors in $X$, and $G_{X}$ is a sparse graph or
is a 2-extension of Petersen or Heawood graph.
\end{lemma}

\begin{proof}
  First, we show that $a$ and $b$ both have  at least two neighbors in $X$. 
  Suppose that one of $a,b$, say $a$,  has a
  unique neighbor $a' \in X$.  We claim that $a'$ is not adjacent to
  $b$.  For otherwise, since $G[X \cup \{a, b\}]$ does not induce a
  path and $G$ is 2-connected, there is a path in $X$ from $b$ to
  $a'$, that together with a path from $a$ to $b$ with interior in $Y$
  form a cycle with a unique chord: $a'b$.  So, $a'$ is not adjacent
  to $b$.  Hence, by replacing $a$ by $a'$, we obtain a proper
  2-cutset that contradicts the minimality of $X$. 
   
  Denote by $m$ the marker of $G_{X}$. One can easily check that
  the block $G_{X}$ is a 2-connected unichord-free graph.  Also,
  $G_X$ is square-free.  Indeed, since $G$ is square-free, a square in
  $G_X$ must be formed by $m, a, b$ and a vertex $x \in X$ adjacent to
  $a$ and $b$.  If $x$ has degree 2, there is a contradiction, because
  from the definition of a block of decomposition, $x$ should have
  been used as the marker vertex.  So, $x$ has a neighbor $x' \in X$.
  Since $G_X$ is 2-connected, in $G_X \setminus x$, there is a path $P
  = x' \dots y$ such that $y$ has a neighbor in $\{a, b\}$.  We choose
  such a path of minimum length.  Since $G$ is square-free, $y$ has
  in fact a unique neighbor in $\{a, b\}$.  Hence, $V(P) \cup \{a, b,
  x, m\}$ is a cycle with a unique chord in $G_X$, a contradiction.    
  
  Suppose $G_X$ has a proper 2-cutset with split $(X_1,X_2,u,v)$.
  Choose it so that $u$ and $v$ both have degree at least 3 (this is
  possible as explained at the beginning of the proof).  Note that
  $m\notin \{u, v\}$.  Observe that, if $\{a,b\}=\{u,v\}$, then
  $(X_1,Y\cup X_2,a,b)$ would be a split of a proper 2-cutset of $G$,
  contradicting the minimality of $|X|$. So
  $\{a,b\}\neq\{u,v\}$. Assume w.l.o.g.\ $b \not \in \{ u,v \}$.

  Suppose $a \not \in \{ u,v \}$.  Then w.l.o.g.\ $\{ a,b \} \subseteq
  X_1$, and hence $(X_1 \cup Y, X_2,u,v)$ --- with $m$ removed if $m$
  is not an original vertex of $G$ --- is a split of a proper 2-cutset
  of $G$, contradicting the minimality of $|X|$.  Therefore $a \in \{
  u,v \}$. Then w.l.o.g.\ $m\in X_{1}$, and hence $(X_1 \cup
  Y,X_2,u,v)$ --- with $m$ removed if $m$ is not an original vertex of
  $G$ --- is a proper 2-cutset of $G$ whose block of decomposition
  $G_{X_2}$ is smaller than $G_{X}$, contradicting the minimality of
  $|X|$.

  In any case, we reach a contradiction which means that $G_X$ has no
  proper 2-cutset.  Since $a$ and $b$ are not adjacent, $G_X$ is not a
  complete graph.  Hence, by Theorem~\ref{t:newdecompositionNoS},
  $G_X$ must be sparse or a 2-extension of Petersen or Heawood graph.\qed
\end{proof}

\section{Total-colouring $\{$square,unichord$\}$-free graphs with maximum degree~3}
\label{s:total-colouring}

In the present section, we prove that the only Type~2
$\{$square,unichord$\}$-free graph of maximum degree 3 is $K_4$.

\begin{theorem}
\label{t:main}
Every $\{$square,unichord$\}$-free graph with maximum degree at most 3
different from $K_4$ is 4-total-colourable.
\end{theorem}

For the proof of Theorem~\ref{t:main}, we need three lemmas ---
Lemmas~\ref{l:extPath},~\ref{l:extSparse} and~\ref{l:extPH} --- that
give sufficient conditions to extend a partial 4-total-colouring of a
graph to a 4-total-colouring of this graph. Lemmas~\ref{l:extPath}
and~\ref{l:extSparse} are proved in~\cite{Chordless}; for the sake of
completeness, all proofs are included here.

\begin{lemma}
  \label{l:extPath}
  Let $k\geq 3$ be an integer, and let  $P= p_1 \dots p_k$ be a path.  Suppose
  that $p_1,  p_1p_2, p_{k-1}p_k, p_k$ are coloured with 2 or 3 colours,
  respectively $c_1$, $c_2$, $c_{2k-2}$ and $c_{2k-1}$, such that
  adjacent elements receive different colours and we do not have 

  $$c_1 = c_{2k-2}  \text{ and }  c_2 = c_{2k-1}.$$  

  This can be extended to a 4-total-colouring of $P$.
\end{lemma}

\begin{proof}
  We view a total-colouring of $P$ as a sequence $c_1, \dots, c_{2k-1}$
  of integers from $\{1, 2, 3, 4\}$, that are the colours of the
  elements of $P$ as they appear when walking along $P$ from $p_1$ to
  $p_k$.  The sequence is \emph{proper} when any two numbers at
  distance at most~2 along the sequence have different values (this
  corresponds exactly to a total-colouring of the path).  What we need
  to prove is that when $c_1, c_2, c_{2k-2}, c_{2k-1}$ receive values
  among $1, 2, 3$, and these values are different for distinct
  elements at distance at most~2, and we do not have $c_1 = c_{2k-2}
  \text{ and } c_2 = c_{2k-1}$, then this can be extended to a proper
  sequence $c_1, \dots, c_{2k-1}$.

  We proceed by induction on $k$.  If $k=3$ then only $c_3$ has no
  value, and the fourth number 4 is available for it.  So suppose
  $k\geq 4$.  Among colours 1, 2, 3, at least one, say 3, is used at
  most once for $c_1, c_2, c_{2k-2}, c_{2k-1}$.  Up to symmetry we may
  assume that colour 3 is not used for $c_1, c_2$.  We put $c_3=3$.
  Now we consider two cases.

  If $\{c_2, c_{2k-2}\} = \{1, 2\}$, say $c_2=2, c_{2k-2}=1$, then we
  must have $c_1=1$ because 3 is not used for $c_1$, so $c_{2k-1} = 3$
  because otherwise we have $c_{2k-1} = 2$ implying $c_1 = c_{2k-2}
  \text{ and } c_2 = c_{2k-1}$ which is forbidden by assumption.  So,
  we put $c_{2k-3} = 2$.  By the induction hypothesis, we complete the
  sequence $c_2, \dots, c_{2k-2}$.
 
  If $\{c_2, c_{2k-2}\} \neq \{1, 2\}$, then we put $c_{2k-3} = 4$.
  We claim that $|\{c_2, c_3, c_{2k-3}, c_{2k-2}\}| \leq 3$.  Indeed,
  since $\{c_2, c_{2k-2}\} \neq \{1, 2\}$, either $c_2 = c_{2k-2}$
  (which proves the claim), or $c_{2k-2} = 3$ (which also proves the
  claim because $c_3=3$).  Also, $c_2 \neq c_{2k-3} = 4$. Hence, by
  the induction hypothesis, we complete the sequence $c_2, \dots,
  c_{2k-2}$.\qed
\end{proof}

\begin{lemma}
  \label{l:extSparse}
  Let $G$ be a 2-connected sparse graph of maximum degree~3.  Graph~$G$
  is 4-total-colourable.  Moreover suppose that $u$ is a vertex of
  degree~2 that has two neighbors $a, b$ of degree~3 and suppose that
  $a$, $b$, $au$, $ub$ receive respectively colours 1, 1, 2, 3.  This
  can be extended to a total-colouring of $G$ using $4$ colours.
\end{lemma}

\begin{proof}
  Note that, because of its maximum degree, $G$ is not a cycle.  Let
  $G'$ be the reduced graph of $G$.  We first total-colour several
  elements of $G'$.  We give to all branch vertices of $G'$ colour~1.
  We edge-colour $G'$ with colours $2, 3, 4$ (up to a relabeling, it
  is possible to give colour $2, 3$ to $au$, $ub$ respectively).  This
  is possible, because $G'$ is bipartite and a classical theorem due
  to K\" onig says that any bipartite graph $H$ is $\Delta(H)$-edge
  colourable.  We extend this to a total-colouring of $G$ by
  considering one by one the branches of $G$ (that edge-wise partition
  $G$ and vertex-wise cover $G$).  Let $P = p_1\dots p_k$, ($k\geq 3$
  since $G$ is sparse) be such a branch.  The following elements are
  precoloured: $p_1, p_1p_2, p_{k-1}p_k, p_k$.  The precolouring
  satisfies the requirement of Lemma~\ref{l:extPath}, so we can extend
  it to~$P$.\qed
\end{proof}

\begin{lemma}
  \label{l:extPH}
  Let $G$ be a 2-connected 2-extension of Petersen
  or Heawood graph. Then $G$ is 4-total-colourable.
  Moreover suppose that $u$ is a vertex of
  degree~2 that has two neighbors $a, b$ of degree~3 and suppose that
  $a$, $b$, $au$, $ub$ receive respectively colours 1, 1, 2, 3.  This
  can be extended to a total-colouring of $G$ using $4$ colours.
\end{lemma}

\begin{proof}
  If $G$ is the Petersen or the Heawood graph, the total-colouring is
  shown on Fig.~\ref{f:basicPetersen}.  

  Suppose first that $G$ is obtained from the Petersen graph by
  deleting exactly one vertex (and then subdividing edges incident to
  at least one degree-2 vertex).  This means that the reduced graph of
  $G$ is $P^*$.  On Figure~\ref{f:PetersenStar}, a 4-total-colouring
  of $P^{*}$ is shown. Note that, for all paths of length~2 of
  $P^{*}$, say $xyz$, the colors of $x,xy,yz,z$ never have the pattern
  $ABAB$. This means that Lemma~\ref{l:extPath} allows to extend this
  4-total-colouring of $P^{*}$ to a 4-total-colouring of $G$.  In
  fact, the total-coloring shown in Figure~\ref{f:PetersenStar} can be
  used for any 2-extension of the Petersen graph (where more than one
  vertex is deleted), because not only the branches of length 2, but
  \emph{all} paths of length 2 in $P^*$ are coloured without using the
  patern $ABAB$.
  
  \begin{figure}[thb]
  \begin{center}
  \includegraphics[width=160pt]{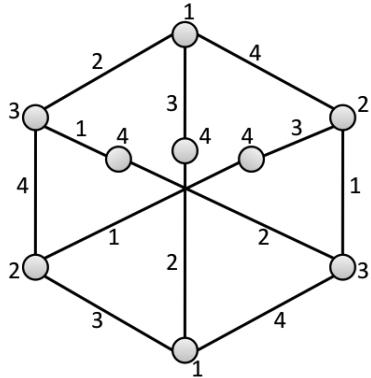}
  \caption{4-total-colouring of the Petersen graph minus one vertex.}
  \label{f:PetersenStar}
  \end{center}
  \end{figure}
  
  Now, suppose that $G$ is obtained from the Heawood graph by deleting
  one vertex (and then subdividing edges incident to at least one
  degree-2 vertex).  This means that the reduced graph of $G$ is
  $H^*$.  On Figure~\ref{f:HeawoodStar}, a 4-total-coloring of $H^{*}$
  is shown. Here, all branches of length~2 have a ``good patern''
  (that is not ABAB), so Lemma~\ref{l:extPath} handles their
  subdivisions.  But unfortunately, some paths of length~2 have a bad
  pattern. So, we are done when exactly one vertex is deleted, but we
  have to study what happens when 2 vertices are deleted.
  
  \begin{figure}[thb]
  \begin{center}
  \includegraphics[width=180pt]{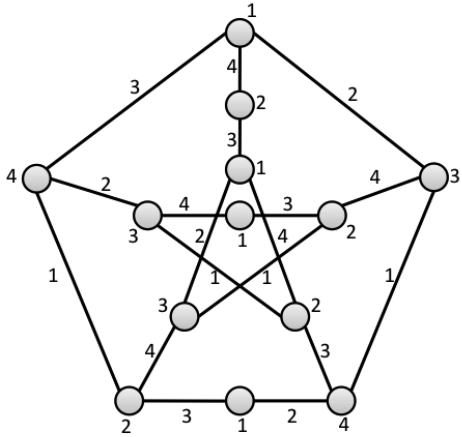}
  \caption{4-total-colouring of the Heawood graph minus one vertex.}
  \label{f:HeawoodStar}
  \end{center}
  \end{figure}
  
  So, suppose that $G$ is obtained from the Heawood graph by deleting
  two vertices (and then subdividing edges incident to at least one
  degree-2 vertex).  Figures~\ref{f:HeawoodStarStarA}
  and~\ref{f:HeawoodStarStarB} show the only two reduced graphs that
  may happen.  All other cases are either isomorphic to these, or have
  a cutvertex.  In the graph of Figure~\ref{f:HeawoodStarStarA}, a
  2-extension of the Petersen graph is obtained, so we are done (to
  see this, consider the reduced graph, and add a vertex adjacent to
  the three vertices of degree 2, this gives a classical embedding of
  Petersen).  In the graph of Figure~\ref{f:HeawoodStarStarB}, a
  coloring is shown.
  
  \begin{figure}[thb]
  \begin{center}
  \includegraphics[width=120pt]{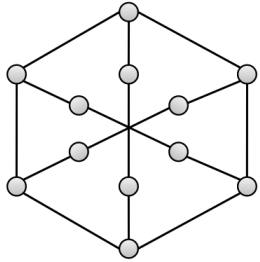}
  \caption{Heawood graph minus two vertices --- a 2-extension of the Petersen graph.}
  \label{f:HeawoodStarStarA}
  \end{center}
  \end{figure}
  
  \begin{figure}[thb]
  \begin{center}
  \includegraphics[width=160pt]{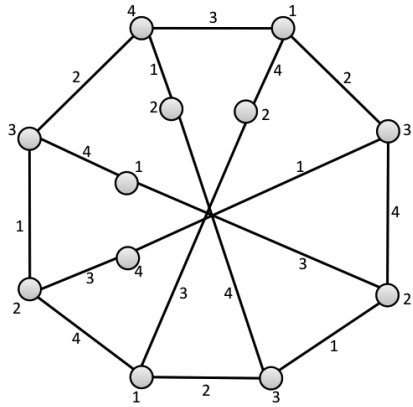}
  \caption{4-total-colouring of the Heawood graph minus two vertices.}
  \label{f:HeawoodStarStarB}
  \end{center}
  \end{figure}
  
  Now suppose that $G$ is obtained by deleting 3 vertices from the
  Heawood graph (and then subdividing edges incident to at least one
  degree-2 vertex).  Since we are done in graph of
  Figure~\ref{f:HeawoodStarStarA} (because all 2-extensions of the
  Petersen graph are already handled), we may assume that one vertex
  is deleted from graph of Figure~\ref{f:HeawoodStarStarB}, and up to
  symmetry, there is only one way to do so. This leads us to the graph
  of Figure~\ref{f:HeawoodStarStarStar} where a coloring is
  shown. Here, up to symmetries, there are two possible places for $u$
  (highlighted in the figure) but the coloring handles both. Note that
  it is essential that all branches of length 2 have a good pattern
  (not ABAB).  If more vertices are deleted, a 2-extension of the
  Petersen graph or a sparse graph is obtained.\qed
  
  \begin{figure}[thb]
  \begin{center}
  \includegraphics[width=160pt]{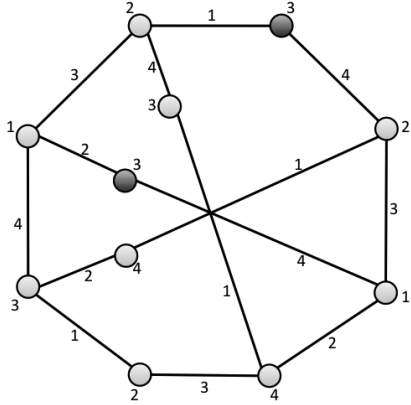}
  \caption{4-total-colouring of the Heawood graph minus three vertices.}
  \label{f:HeawoodStarStarStar}
  \end{center}
  \end{figure}
  
\end{proof}

\subsection*{Proof of Theorem~\ref{t:main}}
\label{s:proof}

\begin{figure}[t]
     \centering
     \subfigure[]{
          \label{f:extTotI1} 
          \includegraphics[width=120pt]{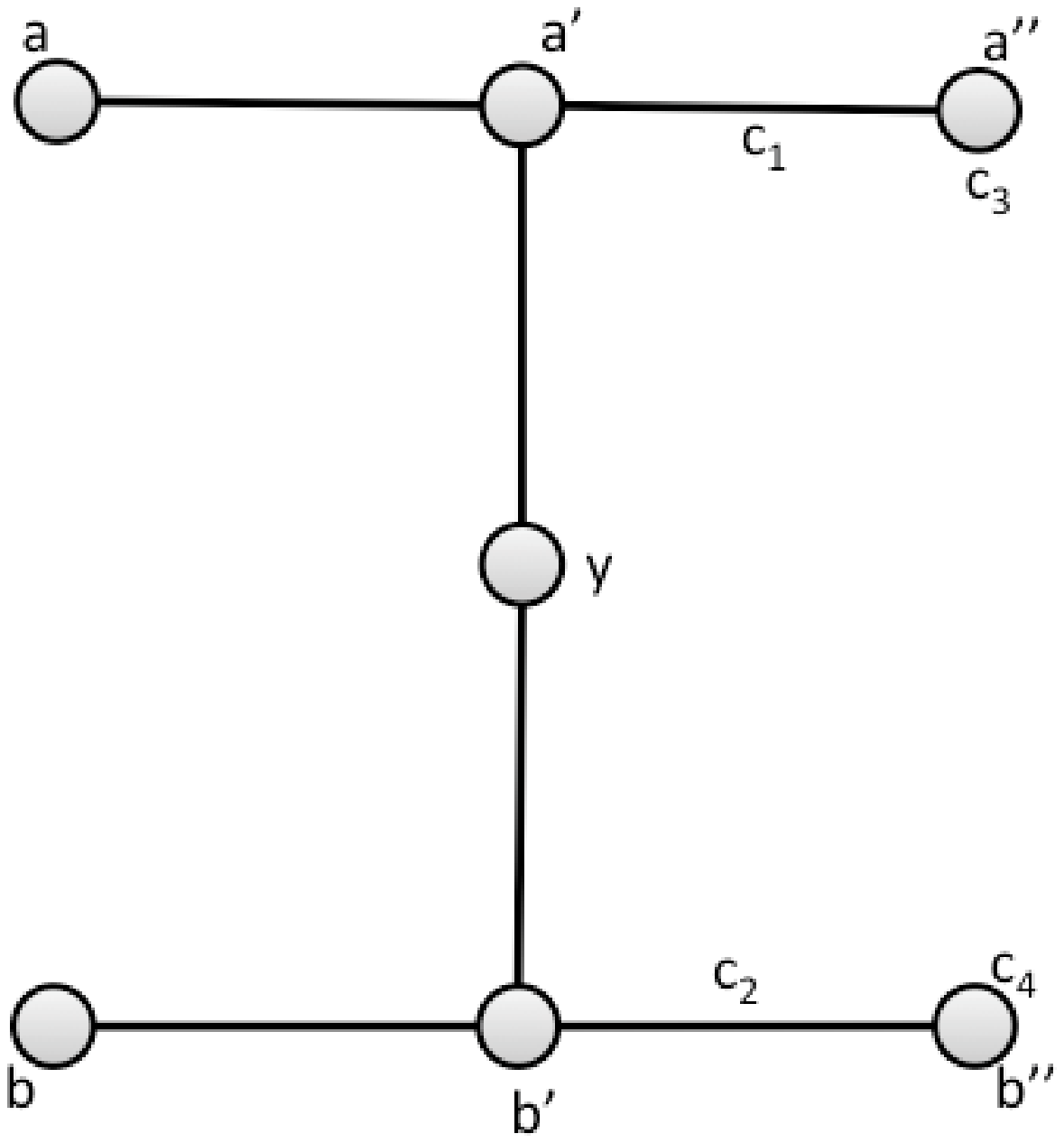}
}
          \hspace{20pt}
     \subfigure[]{
          \label{f:extTotI2} 
          \includegraphics[width=120pt]{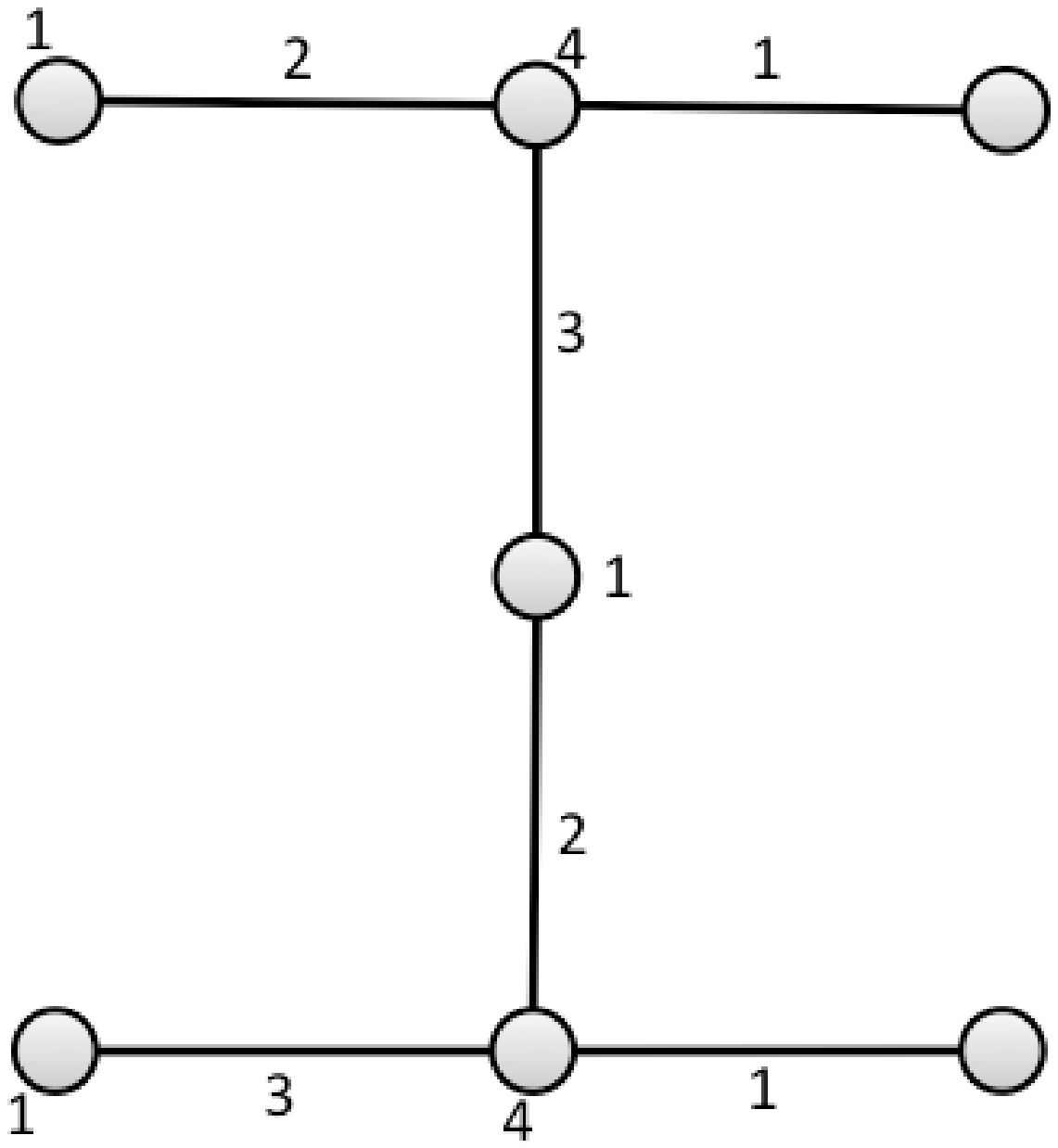}
}
     \subfigure[]{
          \label{f:extTotI3} 
          \includegraphics[width=120pt]{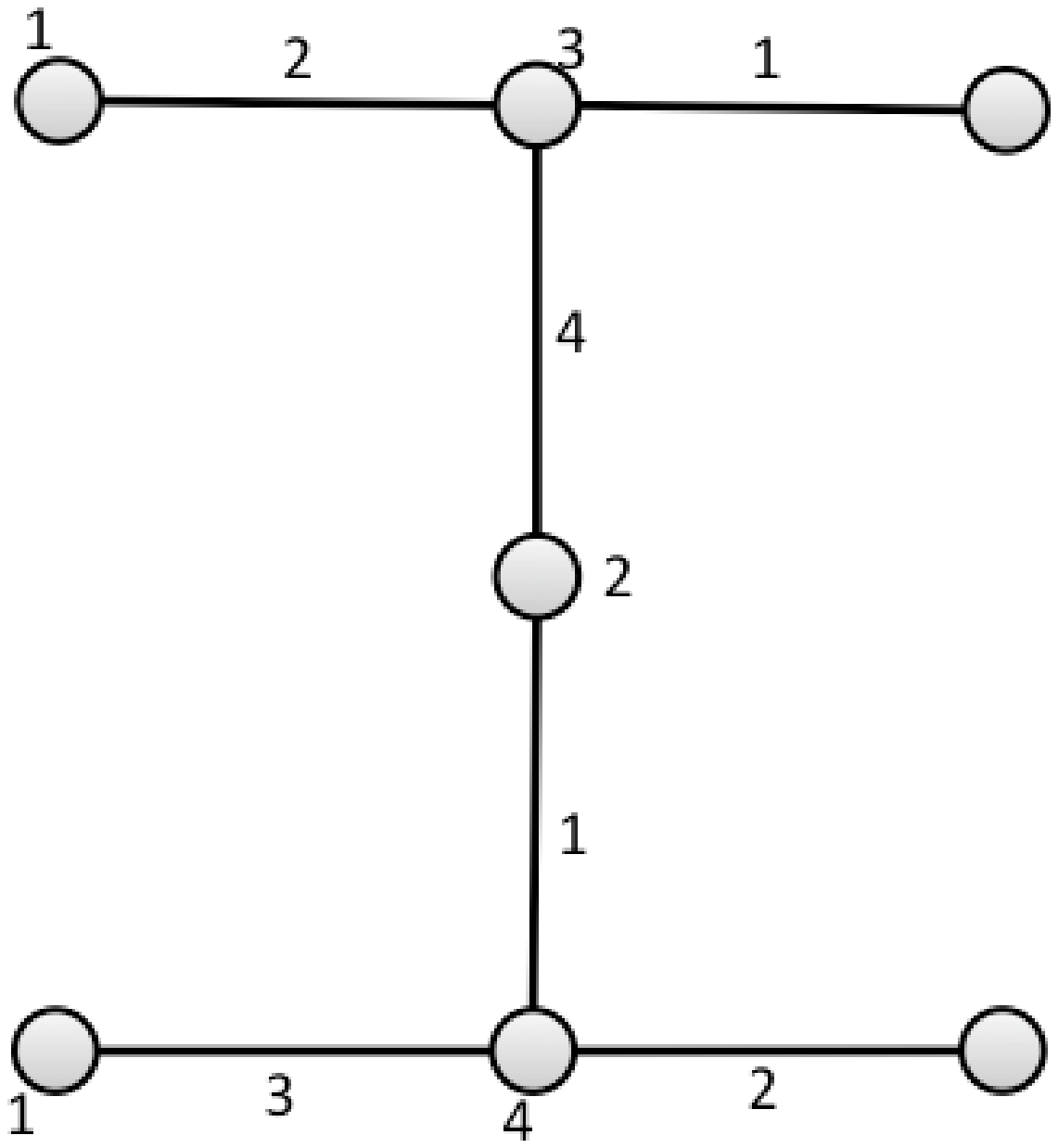}
}
     \subfigure[]{
          \label{f:extTotI4} 
          \includegraphics[width=120pt]{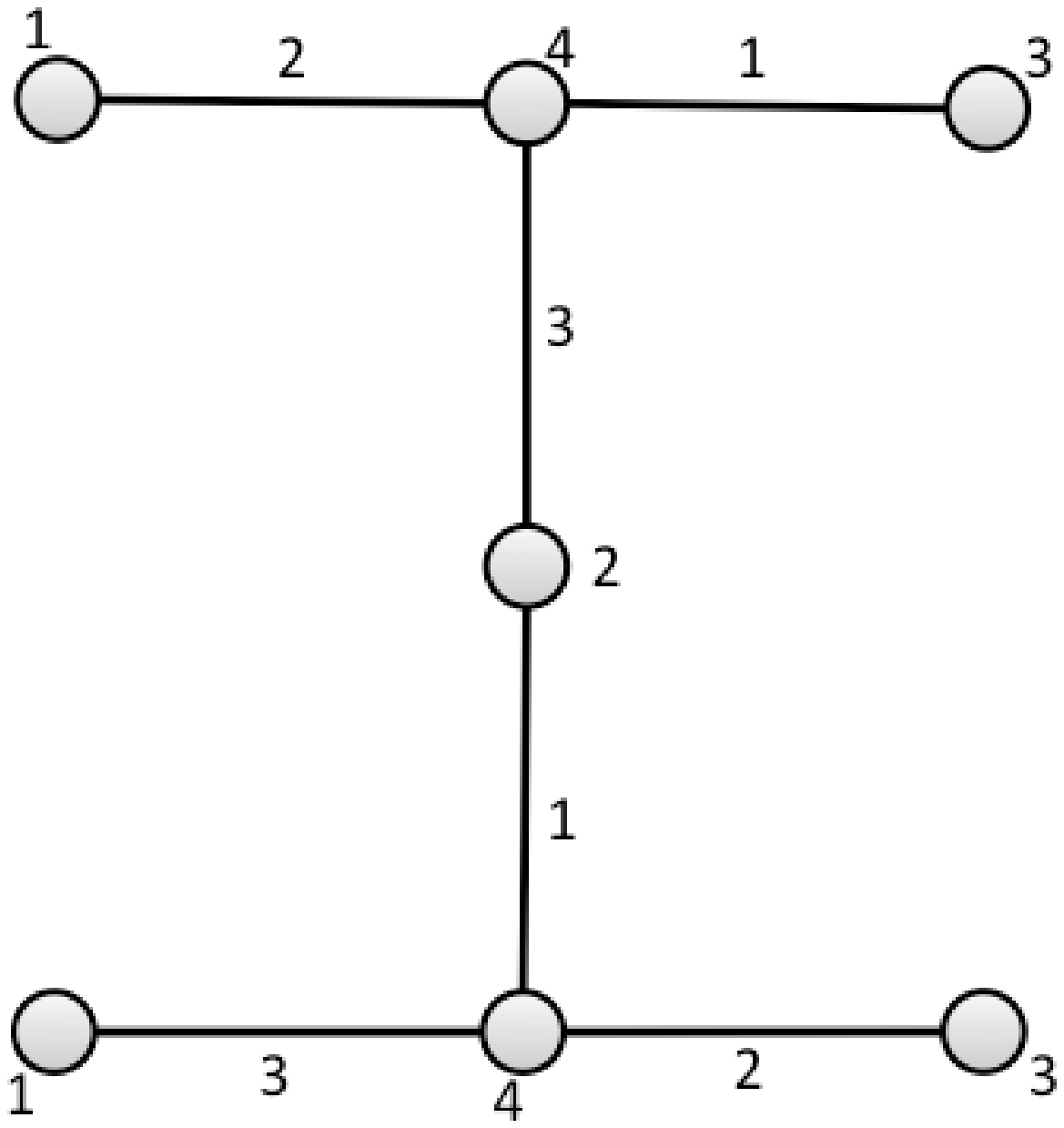}
}
    \caption{Several 4-colourings.}
    \label{f:colouring-basic}
\end{figure}

\begin{proof}
  If the maximum degree of a graph is at most~2, then the graph is a
  disjoint union of paths and cycles, so the conclusion holds.  So,
  let~$G$ be a $\{$square,unichord$\}$-free graph of maximum degree~3
  that is not a complete graph on four vertices.  We shall prove that
  $G$ is 4-total-colourable by induction on $|V(G)|$.  By
  Lemmas~\ref{l:extSparse} and~\ref {l:extPH} this holds for sparse
  graphs and for 2-connected 2-extensions of the Petersen graph and
  the Heawood graph (in particular for the claw=$K_{1,3}$, the
  smallest $\{$square,unichord$\}$-free graph of maximum degree at
  least 3).

  If $G$ has a 1-cutset with split $(X, Y, v)$, a 4-total-colouring of
  $G$ can be recovered from 4-total-colourings of its blocks.  Hence,
  we suppose that $G$ is 2-connected, and we apply
  Theorem~\ref{t:newdecompositionNoS}.  The only outcome not handled
  so far is that $G$ has a proper 2-cutset with split $(X, Y, a, b)$,
  and we choose such a 2-cutset subject to the minimality of $|X|$.
  By Lemma~\ref{l:newextremal}, $a$ and $b$ both have two neighbors in
  $X$ and the block of decomposition $G_X$ is sparse or is a
  2-extension of Petersen or Heawood graph.

  Since $G$ has maximum degree 3, by Lemma~\ref{l:newextremal},
  vertices $a, b$ both have a unique neighbor in $Y$, say $a', b'$
  respectively.  Note that $a'\neq b'$, for otherwise, $a'=b'$ would
  be a cutvertex of the graph (since $|Y| \geq 2$ and $a, b$ have
  degree~3).  We claim that there exists a total-colouring of $G[Y
  \cup \{a, b\}]$ such that $a$ and $b$ both receive colour~1, $aa'$
  receive colour~2 and $bb'$ receive colour~3.  To prove the claim, we
  consider two cases.

  \noindent{\bf Case~1:} $Y$ contains no vertex adjacent to $a'$ and
  $b'$.  This property allows us to build the block of decomposition
  $G_Y$ in a slightly unusual way: block $G_Y$ is obtained from $G[Y
  \cup \{a, b\}]$ by contracting $a$ and $b$ into a new vertex
  $u_{ab}$.  Since $a' \neq b'$, this does not create a double edge.
  Let us check that $G_Y$ is $\{$square,unichord$\}$-free.  Since $a$
  and $b$ have each a unique neighbor in $Y$, it is easily seen that
  $G_Y$ is unichord-free.  Since $G$ is square-free, a square in $G_Y$
  must be formed by $u_{ab}, a', b'$ and a common neighbor $a'$ and
  $b'$, a contradiction to the assumption of Case~1.  Now, by the induction
  hypothesis, we total-colour $G_Y$.  Up to symmetry, $u_{ab}$
  receives colour~1, $u_{ab}a'$ receive colour~2, and $u_{ab}b'$ receive
  colour~3.  This 4-total-colouring is also a 4-total-colouring of $G[Y
  \cup \{a, b\}]$ (we give colour~1 to $a$ and $b$, color 2 to $aa'$, color 3 to $bb'$,
  and the color of any other element in $G[Y\cup\{a,b\}]$ is the same as in $G_Y$).  
  This completes the proof of the claim in Case~1.

  \noindent{\bf Case~2:} $Y$ contains a vertex $y$ adjacent to $a'$
  and $b'$.  
  If both $a'$ and $b'$ have degree~2 then $G$ is sparse.
  If one of $a'$ or $b'$ has degree~2 then this vertex is a 1-cutset of $G$.
  Hence we may assume that both $a'$ and $b'$ have degree~3.
  We total-colour $G[Y \setminus \{y\}]$ by the induction
  hypothesis, and check that vertices $a'$ and $b'$ can be recoloured
  so that the 4-total-colouring can be extended to a 4-total-colouring of
  $G[Y \cup \{a, b\}]$ that satisfies our constraints.  We call $a''$
  (resp.\ $b''$) the unique neighbour of $a'$ in $Y$ that is not $y$
  or $a$ (resp.\ $b$).  We suppose that $a'a''$, $b'b''$, $a''$ and
  $b''$ receive colours $c_1$, $c_2$, $c_3$, $c_4$ respectively, see
  Figure~\ref{f:extTotI1}.

  If $c_1=c_2$, say $c_1 = c_2 = 1$, then $|\{c_1, c_2, c_3, c_4\}|
  \leq 3$, so at least one colour, say 4, is in $\{1, 2, 3, 4\}
  \setminus \{c_1, c_2, c_3, c_4\}$.  We may use this colour to
  recolour $a'$ and $b'$, and the 4-total-colouring can be extended as
  in Figure~\ref{f:extTotI2}.  So, from here on, we suppose $c_1 \neq
  c_2$, say $c_1 = 1$ and $c_2 = 2$.

  If $c_3 \neq c_4$, then up to a relabelling of the colours, we may
  recolour $a'$ and $b'$ with colours $3$ and $4$ respectively.
  Indeed, this can be checked when $\{c_3, c_4\}$ is any of $\{1,
  2\}$, $\{1, 3\}$, \dots, $\{3, 4\}$.  Then the 4-total-colouring can
  be extended as in Figure~\ref{f:extTotI3}.

  Finally, we may assume $c_3 = c_4$, say $c_3 = c_4 = 3$.  In this
  case, the 4-total-colouring can be extended as in
  Figure~\ref{f:extTotI4}.  This completes the proof of the claim in
  Case~2.

  Now, let $G_X$ be defined by adding to $G[X \cup \{a, b\}]$ a vertex
  $u$ adjacent to $a, b$ even if $a$ and $b$ have a common neighbor
  (so this is not the block as defined after Theorem~\ref{th:1}).  
%
%
  If there is a node of $X$ whose neighborhood in $G$ is $\{a,b\}$
  then $G[X\cup\{a,b\}]$ is
  a cycle (because of the minimality of $X$) and $G_X$ is sparse, and
  contains a square.  Otherwise,
%
%
  $G_X'=G_X$,
  so by
  Lemma~\ref{l:newextremal}, $G_X$ is sparse or a 2-extension of
  Petersen or Heawood graph.  In $G_X$, precolour $a$ and
  $b$ with colour~1, $ua$ with colour~2, and $ub$ with colour~3.
  Apply Lemma~\ref{l:extSparse} (if $G_X$ is sparse) or
  Lemma~\ref{l:extPH} (if $G_X$ is a 2-extension of Petersen
  or Heawood graph) to~$G_X$.  This gives a 4-total-colouring of
  $G_X$.  A 4-total-colouring of $G$ is obtained as
  follows: elements of $G$ that are in $G_X$ receive the colour they
  have in $G_X$, and elements that are in $G[Y \cup \{a, b\}]$ receive
  colours as in the claim above.\qed
\end{proof}

\section*{Acknowledgements}

We are deeply indebted to an anonymous referee whose comments and guidance helped us to improve a lot our manuscript.
This research was partially supported by the Brazilian research agencies CNPq (grant Universal 472144/2009-0) and FAPERJ (grant INST E-26/111.837/2010).
The third author is partially supported by the French \emph{Agence Nationale de la Recherche} under reference \textsc{anr 10 jcjc 0204 01}.

\newpage




\end{document}